\documentclass{article}

\usepackage[utf8]{inputenc}
\usepackage[T1]{fontenc}
\usepackage{graphicx}
\usepackage{fullpage}
\usepackage[normalem]{ulem}
\usepackage{amsmath}
\usepackage{amssymb}
\usepackage{amsthm}
\usepackage{hyperref}
\usepackage{comment}
\usepackage{authblk}
\usepackage[algo2e,linesnumbered]{algorithm2e}
\usepackage{xcolor}

\newtheorem{lemma}{Lemma}[section]
\newtheorem{theorem}[lemma]{Theorem}

\newtheorem{definition}[lemma]{Definition}

\newcommand{\email}[1]{\href{mailto:#1}{#1}}

\author[1]{Shreyas Pai} 
\author[2]{Sriram V.~Pemmaraju}
\affil[1]{Aalto University, Finland. Email: \email{shreyas.pai@aalto.fi}}
\affil[2]{University of Iowa, USA. Email: \email{sriram-pemmaraju@uiowa.edu}}
\date{}

\DeclareMathOperator*{\E}{\mathbf{E}}

\newcommand{\eps}{\varepsilon}
\newcommand{\poly}{\operatorname{\text{{\rm poly}}}}
\newcommand{\congest}{\textsc{Congest}}
\newcommand{\local}{\textsc{Local}}

\newcommand{\pr}[1]{\Pr\! \left[ {#1} \right]}

\newcommand{\cc}{\textsc{CongClique}}
\newcommand{\ot}{\tilde{O}}
\newcommand{\sparsify}{\textsc{Sparsify}}

\newif\ifdraft%
\newif\ifshort%

\shorttrue
\drafttrue%

\begin{document}

\title{Deterministic Massively Parallel Algorithms for Ruling Sets}
\maketitle

\begin{abstract}
In this paper we present a \textit{deterministic} $O(\log\log n)$-round algorithm for the $2$-ruling set problem in the Massively Parallel Computation model with $\tilde{O}(n)$ memory; this algorithm also runs in $O(\log\log n)$ rounds in the Congested Clique model.
This is exponentially faster than the fastest known deterministic 2-ruling set algorithm for these models, which is simply the $O(\log \Delta)$-round deterministic Maximal Independent Set algorithm due to Czumaj, Davies, and Parter (SPAA 2020).
Our result is obtained by derandomizing the 2-ruling set algorithm of Kothapalli and Pemmaraju (FSTTCS 2012).
\end{abstract}

\section{Introduction}
There has been substantial progress recently on derandomizing distributed algorithms in ``all-to-all'' communication models such as Congested Clique (\cc{}) and Massively Parallel Computation (MPC).
A recent example is the deterministic maximal independent set (MIS) algorithm due to Czumaj, Davies, and Parter \cite{Czumaj2019}.
By derandomizing the well-known randomized MIS algorithm of Luby \cite{LubySICOMP1986} and Alon, Babai, and Itai \cite{AloneBIJAlg1986} they
obtain an $O(\log \Delta +\log\log n)$-round \emph{deterministic} MIS algorithm in the MPC model, with $O(n^{\eps})$ space on each machine for any constant $\eps > 0$.
This algorithm runs in $O(\log \Delta)$ rounds in the MPC model with $\ot(n)$ memory and in the \cc{} model.

Continuing this trend, in this paper we present a \textit{deterministic} $O(\log\log n)$-round algorithm for the 2-ruling set problem in the MPC model with $\ot(n)$ memory; this algorithm also runs in 
$O(\log\log n)$ rounds in the \cc\ model.
A \textit{2-ruling set} of a graph $G = (V, E)$ is an independent set $I \subseteq V$ such that every vertex in $V$ is at most 2 hops away from some vertex in $I$.
A 2-ruling set is a natural relaxation of an MIS, which can be equivalently viewed as a 1-ruling set.
The fastest known \textit{deterministic} 2-ruling set algorithm in the MPC model with $\ot(n)$ memory (or in the \cc\ model) is simply the aforementioned deterministic MIS algorithm~\cite{Czumaj2019}, that runs in $O(\log \Delta)$ rounds.
If randomness is permitted, 2-ruling sets can be solved in $O(\log\log\log n)$ rounds w.h.p.\footnote{We use ``w.h.p.~as short for ``with high probability'' and it refers to probability at least $1-n^{-c}$ for constant $c \ge 1$.} in both the MPC model with $\ot(n)$ memory and in the \cc{} model~\cite{HegemanPSDISC2014,HegemanPSarxiv2014}.
Thus our result can be viewed as an exponential improvement on the fastest known deterministic algorithm and an exponential closing of the gap between randomized and deterministic 2-ruling set algorithms in these models.

The starting point for this deterministic algorithm is the randomized 2-ruling set algorithm of Kothapalli and Pemmaraju~\cite{kothapalli12_super_fast_rulin_sets} that runs in the \local\ model. In the key step of this algorithm, called \sparsify, a random vertex subset $S$ is repeatedly chosen such that w.h.p.~(i) the subgraph induced by $S$ is sparse and (ii) the neighborhoods of all high degree vertices are ``hit'' by $S$. Our main technical contribution is to show that the
randomness needed in this step can be reduced to \(O(\log n)\) bits.
Subsequently, we show that it is possible to deterministically set the ``random'' bits in $O(1)$ rounds by using the method of conditional expectations, specifically the 
derandomization framework developed by Censor-Hillel, Parter, and Schwartzman~\cite{Censor-HillelPS16}.

\subsection{The MPC and \cc{} Models}\label{sec:MPC}
The \textit{Massively Parallel Computing (MPC)} model was developed and refined in a sequence of papers~\cite{FeldmanMSSSTALG2010,GoodrichSZISAAC2011,karloff2010,BeameKSJACM2017}.
It is defined by a set of machines, each having at most $S$ words of memory. The machines are connected to each other via an all-to-all communication network.
Communication and computation in this model are synchronous.
In each round, each machine receives up to $S$ words from other machines, performs local computation, and sends up to $S$ words to other machines. 
The key characteristic of the MPC model is that both the memory upper bound $S$ and the number of machines used are assumed to be strongly sublinear in the input size $N$, i.e., bounded by $O(N^{1-\delta})$ for some constant $\delta$, $0 < \delta < 1$.
This characteristic models the fact that in modern large-scale computational problems the input is too large to fit in a single machine and is much larger than the number of available machines. 
Thus, for graph problems, where the input is an $n$-vertex, $m$-edge graph, $S$ is bounded by $O((m+n)^{1-\delta})$.
While $S$ is sublinear in the total input size, how it relates to the number of vertices $n$ has
a significant impact on the efficiency of algorithms we are able to design.
Researchers have focused on three regimes: (i) superlinear memory: $S = O(n^{1+\eps})$, for $\eps > 0$, (ii) linear memory: $S = \tilde{O}(n)$, and (iii) sublinear memory: $S = O(n^\eps)$ for $\eps > 0$.
In this paper, we work in the \emph{linear memory MPC} model.

The \cc\ model is a classical model of distributed computation introduced by Lotker,
Pavlov, Patt-Shamir, and Peleg \cite{LotkerPPPSPAA2003}. 
If the input is a graph $G$ with $n$ nodes, the goal is to solve some graph problem on $G$ by 
performing computation at the nodes of $G$. 
Computation and communication proceeds in synchronous rounds and in each round every node can 
send an $O(\log n)$-bit message to every node in $G$, not just the nodes it is adjacent to. Thus the
underlying communication network is $K_n$.

\subsection{Other Related Work}\label{sec:related-work}
Besides the deterministic MIS algorithm mentioned earlier~\cite{Czumaj2019}, Czumaj, Davies, and Parter~\cite{CzumajDPPODC2020} give a deterministic constant round $(\Delta + 1)$-coloring algorithm in \cc{}.
Censor-Hillel, Parter, and Schwartzman~\cite{Censor-HillelPS16} introduced a distributed derandomization framework, based on the method of conditional expectations~\cite{Raghavan1988}, in \congest\ and \cc\ which allowed them to obtain a deterministic $O(\log \Delta \cdot \log n)$ round algorithm for MIS in \cc.
Our overall approach fits in this framework.
They also showed how to obtain a $(2k-1)$-spanner with $O(k n^{1+1/k}\log n)$ edges deterministically in $O(k \log n)$ rounds in \cc. 
Parter and Yogev \cite{ParterYogevDISC2018} improve on the spanner results by giving a deterministic construction of a $(2k-1)$-spanner with $\ot(n^{1+1/k})$ edges in $O(\log k+(\log\log n)^3)$ rounds and a deterministic construction of an $O(k)$-spanner with $O(k n^{1+1/k})$ edges in $O(\log k)$ rounds.
At a high level, the approach of Parter and Yogev is to derandomize a sparse hitting set construction.
The \sparsify\ step in the Kothapalli-Pemmaraju 2-ruling set algorithm can also be viewed as a randomized construction of a sparse hitting set.
Due to this connection, the Parter-Yogev approach provides inspiration for this work.

\subsection{Technical Preliminaries}\label{sec:preliminaries}
A key part of our approach is to reduce the randomness needed in the \sparsify\ step of the Kothapalli-Pemmaraju 2-ruling set algorithm. We achieve this by using a $k$-wise independent family of hash functions. In this section we provide some technical preliminaries on concentration bounds for the sum of $k$-wise independent random variables and the number of random bits needed to construct a family of $k$-wise independent hash functions.
The following tail inequality is from Bellare and Rompel~\cite{BellareRSFCS1994}.
\begin{lemma}
\label{lemma:Chernoff2}
Let $k \ge 4$ be an even integer. Suppose
$Z_1, Z_2, \ldots, Z_t$ are $k$-wise
independent random variables taking values in
$[0, 1]$. Let $Z = \sum_{i=1}^t Z_i$ and
$\mu = \mathbb{E}[Z]$, and let $\lambda > 0$.
Then,
$$\pr{|Z - \mu|\ge \lambda} \le 8\left(\frac{k\mu + k^2}{\lambda^2}\right)^{k/2}.$$
\end{lemma}
\noindent
The following is a paraphrased version of Definition 3.31 in Vadhan's work~\cite{VadhanFTTCS2012}.
\begin{definition}
For $N$, $M$, $k \in \mathbb{N}$, such that $k \le N$, a family of functions $\mathcal{H} = \{h: [N] \to [M]\}$ is $k$-wise independent if for all distinct $x_1,x_2,\ldots,x_k
\in [N]$, the random variables $h(x_1), h(x_2), \ldots, h(x_k)$ are independent and uniformly distributed in
$[M]$ when $h$ is chosen uniformly at random from
$\mathcal{H}$.
\end{definition}
\noindent
The following lemma appears as Corollary 3.34 in Vadhan's work~\cite{VadhanFTTCS2012}.
\begin{lemma}
\label{lemma:randomBits}
For every $n,m,k$, there is a family of $k$-wise independent hash functions $\mathcal{H} = \{h:\{0,1\}^n \to \{0,1\}^m\}$
such that choosing a random function from $\mathcal{H}$
takes $k\cdot \max\{n, m\}$ random bits, and evaluating a function from $\mathcal{H}$ takes $\text{poly}(n, m, k)$
computation.
\end{lemma}

\paragraph*{Notation:} In the rest of the paper we use $N(v)$ to denote the set of neighbors of a node $v$, $N^{+}(v)$ to denote $\{v\} \cup N(v)$, $N^{+}(S)$ to denote $\cup_{v \in S} N^{+}(v)$ for any subset $S$ of nodes, and $E(G)$ denote the set of edges in $G$.

\section{Deterministic Ruling Set Algorithm}\label{sec:det-rs-mpc}
We obtain our result in three steps: (i) we use limited dependence in order to
reduce the amount of randomness used by a key sampling step in the Kothapalli-Pemmaraju 2-ruling set algorithm~\cite{kothapalli12_super_fast_rulin_sets}, (ii) we then use the method of conditional expectations, following ideas of Censor-Hillel, Parter, and Schwartzman~\cite{Censor-HillelPS16}, and derandomize this sampling step, and (iii) we show that with this derandomized sampling step in place, the Kothapalli-Pemmaraju 2-ruling set algorithm can be implemented as a deterministic algorithm in the linear memory MPC model and in the \cc\ model in $O(\log\log n)$ rounds.

\subsection{Reducing Randomness via Limited Dependence}
Let \(G(V, E)\) be a graph with \(n\) vertices, \(m\) edges, and maximum degree \(\Delta\). If we sample each vertex with probability \(\alpha/\sqrt{\Delta}\), where $\alpha$ is a small enough constant, then the subgraph induced by the sampled vertices has \(O(n)\) edges in expectation and
every vertex with degree at least \(\sqrt{\Delta}\log n\) has a sampled vertex in its neighborhood w.h.p.
In the MPC model with linear memory, we can gather the sampled subgraph at a single machine and compute an MIS \(I\) of this subgraph.  Therefore, w.h.p.~the MIS \(I\) of the sampled subgraph forms a \(2\)-ruling set of the sampled vertices plus the vertices with degree at least \(\sqrt{\Delta}\log n\). Thus, we can deactivate all the high degree vertices and work with a graph of maximum degree \(\sqrt{\Delta}\log n\).

Our goal in this section is to show how to reduce the randomness used by this sampling step. 
In order to do this, we assume that the every node knows its color $c_u$ in a proper node-coloring of palette size $O(\Delta^3)$.
Node \(u\) applies an appropriate random hash function on its color \(c_u\) to make its random choice.
More specifically, nodes choose a random hash function $h_s$ from the family
$$\mathcal{H} = \{h_{s} \colon [O(\Delta^3)] \to [f] \mid s \in \{0,1\}^{r} \}$$
of $k$-wise independent hash functions, where $k = O(\log n/\log \Delta)$.
Here $1 \le f \le \sqrt{\Delta}$ is a parameter and $r$ is the length of the random binary seed $s$ that defines $h_s$. 
Note that the domain of $h_s$ is the palette of colors assigned to nodes. Thus, for any node $u$, $h_s(u)$ takes on a value in $[f]$ uniformly at random. This implies that picking nodes $u$ with $h_s(u) = 1$ is equivalent to sampling nodes with probability $1/f$. 
To obtain an $O(\Delta^3)$-coloring we apply one iteration of the color reduction technique of Linial~\cite{linial92_local_distr_graph_algor} which converts a proper coloring of \(\kappa\) colors to an \(O(\Delta^2 \log \kappa)\)-coloring.
We can use this color reduction technique on the coloring induced by the ID's to get an a \(O(\Delta^2 \log n) = O(\Delta^3)\)-coloring in one round, assuming \(\Delta \ge \log n\).

The following lemma establishes the randomness reduction result we need. It essentially shows that the subset of nodes sampled using $O(\log n)$ random bits, one can sample a random hash function $h$ from the family $\mathcal{H}$ mentioned above. This hash function can in turn be used to sample a subset $Z_h$ of nodes and $Z_h$ has the properties we need: (i) it is sparse and (ii) it ``hits'' all large enough neighborhoods.

\begin{lemma}\label{lem:MPC-RS-seed}
Let \(G(V, E)\) be a graph with \(n\) vertices, \(m\) edges, and maximum degree \(\Delta = \Omega(\log^{4/(3\eps)} n)\) for a parameter \(0 \le \eps \le 1\). Let \(\mathcal{C} = \{c_u  \mid  u \in V\}\) be a proper \(O(\Delta^3)\)-coloring of \(G\).
Then, for a parameter $f$, \(1 \le f \le \sqrt{\Delta}\) and any constant $c > 0$, there exists a family of \(O(\log n/\log\Delta)\)-wise independent hash functions \(\mathcal{H} = \{h_{s} \colon [O(\Delta^3)] \to [f] \mid s \in \{0,1\}^{r} \}\) such that choosing a random function \(h \in \mathcal{H}\) takes \(r=O(\log n)\) random bits and for \(Z_h = \{u \in V \mid h(c_u)=1 \}\) it holds that:
\begin{enumerate}
  \item \(\E[|E(G[Z_h])|] \le \beta m/f^2\) for some constant $\beta > 0$, and
  \item For all \(v \in V \setminus Z_h\) such that \(|N(v)| \ge f \Delta^{\eps}\), \(\pr{ N(v) \cap Z_h \ne \emptyset } \ge 1-1/n^c\).
\end{enumerate}
\end{lemma}

\begin{proof}
Since we are using a family of \(k\)-wise independent hash functions $\mathcal{H} = \{h: [N] \to [M]\}$
for $k = O(\log n/\log \Delta)$, $N = O(\Delta^3)$, and $M = f \le \sqrt{\Delta}$, Lemma \ref{lemma:randomBits} tells us that $r = k\cdot\min\{O(\log M), O(\log N)\} = O(\log n)$ random bits suffice to pick a random hash function $h$ from $\mathcal{H}$.

First we prove that \(\E[|E(G[Z_h])|] \le O(m/f^2)\). Let \(X_e\) be a random variable indicating whether the edge \(e \in E\) belongs to \(E(G[Z_h])\) or not. For an edge \(e = (u, v)\), we have \(X_e = 1\) if and only if both end points \(u,v \in Z_h\). Since \(h\) belongs to a family of \(O(\log n/\log\Delta)\)-wise independent hash functions, the random bits output by \(h\) are at least pairwise independent. Therefore, \(\pr{X_e = 1} \le \pr{h(c_u) = h(c_v) = 1} = \pr{h(c_u) = 1} \cdot \pr{h(c_v) = 1} \le O(1/f^2)\). Therefore by linearity of expectations, \(\E[|E(G[Z_h])|] = \sum_{e \in E} \E[X_e] \le O(m/f^2)\).

Now we prove the second part of the lemma. Consider a vertex \(v \in V \setminus Z_h\) such that \(|N(v)| \ge f \Delta^\eps\). Let \(Y_u\) be a random variable indicating whether \(u \in Z_h\) or not. And let \(Y\) be the random variable denoting the number of neighbors of \(v\) that are in \(Z_h\). If none of the neighbors of \(v\) belong to the set \(Z_h\), then it means \(Y = \sum_{u \in N(v)} Y_u = 0\). Let \(\mu = \E[Y] = \sum_{u \in N(v)}\E[Y_u] \ge \sum_{u \in N(v)}\pr{Y_u = 1} \ge O(|N(v)|/f) \ge \Delta^{\eps}\). Since \(h\) belongs to a family of \(k = (32c/\eps)\log n/\log\Delta\)-wise independent hash functions we can use limited independence Chernoff bound from Lemma~\ref{lemma:Chernoff2} to show concentration around the expected value.
\begin{align*}
\pr{|Y - \mu|\ge \mu/2} &\le 8\left(\frac{4k\mu + 4k^2}{\mu^{2}}\right)^{k/2} \le 8 \left(\frac{4 \cdot 32c\log n}{\eps\log\Delta \cdot \Delta^\eps} + \frac{4 \cdot (32c\log n)^2}{(\eps\log\Delta \cdot \Delta^{\eps})^2}\right)^{k/2} \\
\therefore \qquad \pr{|Y - \mu|\ge \mu/2} &\le 8\left(\Delta^{-\eps/4}\right)^{\frac{16c\log n}{\eps\log \Delta}} \le 8 \cdot 2^{-4c\log n} \le \frac{1}{n^{2c}}
\end{align*}

Where the third inequality holds because we assume \(\Delta =  \Omega(\log^{4/(3\eps)} n)\) which implies \(\log n = O(\Delta^{3\eps/4})\). Therefore, we get \(\pr{N(v)\cap Z_h = \emptyset} = \pr{Y = 0} \le \pr{ Y \le \Delta^{\eps}/2} \le 1/n^{2c}\). A union bound over all such vertices \(v\) proves the second property and hence, the lemma.
\end{proof}

\subsection{Applying the Method of Conditional Expectations}
Note that even though we showed that \(Z_{h}\) can be sampled using a random seed of size \(O(\log n)\) random bits, this by itself does not imply a deterministic algorithm.
A trivial way to derandomize the sampling of $Z_h$ is to consider all possible values of the seed and pick the one that gives a set with the required properties; such a setting is guaranteed to exist.
But, there are polynomially many choices, so this is not an efficient approach.
In the proof of the following theorem we use the method of conditional expectations and follow the scheme of~\cite{Censor-HillelPS16} for implementing this method in the MPC and \cc\ models.

\begin{theorem}
\label{thm:generalderand}
Let \(G(V, E)\) be a graph with \(n\) vertices, \(m\) edges, and maximum degree \(\Delta = \Omega(\log^{4/(3\eps)} n)\) for parameter \(\eps > 0\). And let \(1 \le f \le \sqrt{\Delta}\) be another parameter. We can design a deterministic algorithm \(\mathcal{A}_{det}\) that constructs a set \(Z \subseteq V\) such that \(|E(G[Z])| \le O(m/f^2)\) and for all vertices \(v \in V \setminus Z\) such that \(|N(v)| \ge f \Delta^\eps\), $N(v) \cap Z \ne \emptyset$.
This set \(Z\) can be constructed in \(O(1)\) rounds of the linear memory MPC and \cc{} models.
\end{theorem}
\begin{proof}
We describe the algorithm \(\mathcal{A}_{det}\) in linear memory MPC; extending it to \cc{} using Lenzen's routing protocol~\cite{lenzen13_optim_deter_routin_sortin_conges_clique} is straightforward.
We first define some random variables.
Let \(E_A\) be the random variable denoting the number of edges \(G[Z_h]\) where \(Z_h\) is the set of nodes constructed in Lemma~\ref{lem:MPC-RS-seed}.
For node \(u \in V\) let \(X_u\) be indicator random variable for the event that \(|N(u)| \ge f \Delta^\eps\) and \(N^{+}(u) \cap Z_h=\emptyset\). For all \(u \in Z_h\) we have \(\pr{X_u=1} = 0\) and by Lemma~\ref{lem:MPC-RS-seed}, for all \(u \in V \setminus Z_h\) we have \(\pr{X_u=1}\leq 1/n^c\).
Define \(\Psi = E_A + n^{4}\sum_{u \in V} X_u\).
We also have \(\E[\Psi] = \E[E_A] + n^{4}\sum_u \pr{X_u=1} \le \beta m/f^2 + n^{4-c} = O(m/f^2)\).

Lemma \ref{lem:MPC-RS-seed} tells us that the random variables $E_A$ and $X_u$ (and hence $\Psi$) are determined
by an $r = O(\log n)$ bit random seed. 
We will use the method of conditional expectations to deterministically set these $r$ ``random'' bits in $t$ chunks,
$Y_1, Y_2, \ldots, Y_t$, of \(\lfloor \log n \rfloor\) bits each, for $t = O(1)$ such that
$\E[\Psi \mid Y_1, Y_2, \ldots, Y_i] \le \E[\Psi \mid Y_1, Y_2, \ldots, Y_{i-1}]$ for $i =1, 2, \ldots, t$.
Thus, $\E[\Psi \mid Y_1, Y_2, \ldots, Y_t] \le \E[\Psi] = O(m/f^2)$.
But, $\Psi$ conditioned on $Y_1, Y_2, \ldots, Y_t$ is completely determined and its value can be bounded
above by $O(m/f^2)$ only when $E_A = O(m/f^2)$ and $X_u = 0$ for all $u$.

We will now show how to compute the assignment of each chunk of the seed in \(O(1)\) rounds. Since
$t = O(1)$, we get an \(O(1)\) round algorithm.
Assume that the assignment for the previous \(i-1\) chunks \(Y_1\ldots,Y_{i-1}\) have been fixed and we're now setting the bits in $Y_i$.
For each of the (at most \(n\)) possible assignments to \(Y_i\), we assign a unique node \(v\) that is responsible for computing \(\E[\Psi  \mid  Y_1, \ldots,Y_{i}]\).

For each \(u \in V\) let \(E_u\) be a random variable denoting the number of sampled edges incident on \(u\) (note that \(E_u = 0\) if \(u\) is not sampled). For a particular value \(Y_i\), the assigned node \(v\) receives the values \(\E[E_u \mid Y_1,\ldots,Y_{i}]\) from all \(u \in V\) and \(\pr{X_u=1 \mid Y_1,\ldots,Y_{i}}\) from all nodes \(u \in V\). Notice that a node \(u\) can compute the conditional values \(\E[E_u \mid Y_1,\ldots,Y_{i}]\) and \(\pr{X_u=1 \mid Y_1,\ldots,Y_{i}}\), since \(u\) knows the IDs of the vertices in \(N(u)\) and has all the information for this computation.
The node \(v\) then computes \(\sum_{u \in V}{\E[E_u \mid Y_1,\ldots,Y_{i}]}\) and \(\sum_{u \in V}{\pr{X_u = 1  \mid  Y_1,\ldots,Y_{i}}}\) and sends them to a global leader \(w\). Thus, using the values received from all the assigned nodes, \(w\) knows \(\sum_{u \in V}{\E[E_u \mid Y_1,\ldots,Y_{i}]}\) as well as \(\sum_{u \in V}{\pr{X_u = 1  \mid Y_1,\ldots,Y_{i}}}\) for all of the possible \(n\) assignments to \(Y_i\).
Finally, \(w\) chooses the \(Y_{i}\) that minimizes \(\sum_{u \in V}{\E[E_u \mid Y_1,\ldots,Y_{i}]} + \) \(\sum_{u \in V}{\pr{X_u = 1  \mid Y_1,\ldots,Y_{i}}}\) and broadcasts this choice to all nodes in the graph. Note that this implies after \(O(1)\) rounds, the seed has been completely fixed and this is the good seed we wanted to compute. We get a good hash function \(h \in \mathcal{H}\) from Lemma~\ref{lem:MPC-RS-seed}, which gives us the set of nodes \(Z = Z_h\) that satisfies the theorem.
\end{proof}

\subsection{Deterministic 2-Ruling Sets}
The 2-ruling set algorithm of Kothapalli and Pemmaraju \cite{kothapalli12_super_fast_rulin_sets}, with the random sampling step replaced by a deterministic step with the necessary properties is shown below.

\RestyleAlgo{boxruled}
\begin{algorithm2e}\caption{\textsc{Deterministic-2-Ruling-Set}\((G)\)\label{alg:2rsmpc}}
\(U \leftarrow \emptyset\), \(\Delta \leftarrow \Delta(G)\) \tcp{Initially ruling set is empty and all nodes are active}
\While{\(\Delta = \Omega(\log^4 n)\)} {
    Compute an \(O(\Delta^3)\) coloring using one iteration of Linial's coloring algorithm~\cite{linial92_local_distr_graph_algor} \\
    Using Theorem~\ref{thm:generalderand} with parameters \(f = \sqrt{\Delta}\) and \(\eps = 1/3\) deterministically compute a set \(Z\) \\
    Send the subgraph induced by \(Z\) to a single machine \(w\) which computes an MIS \(I\) of \(G[Z]\) \\
    Let \(H\) be the set of nodes in \(G\) with degree at least \(f\Delta^{\eps} = \Delta^{5/6}\) \\
    \(G \leftarrow G[V \setminus (N^{+}(I) \cup H)]\) \tcp{Nodes with a neighbor in \(I\) and high degree nodes are deactivated}
    \(\Delta \leftarrow \Delta(G)\), \(U \leftarrow U \cup I\) \\
}
Run the deterministic low-degree MIS algorithm of~\cite{Censor-HillelPS16} on \(G\) and add the MIS nodes to \(U\) \\
\Return\ \(U\)
\end{algorithm2e}

\begin{theorem}\label{thm:det-rs}
Algorithm~\ref{alg:2rsmpc} deterministcally computes a \(2\)-ruling set of \(G\) in \(O(\log \log n)\) rounds of the linear memory MPC and \cc{} models.
\end{theorem}
\begin{proof}
Running one iteration of Linial's coloring algorithm in MPC and \cc{} requires \(O(1)\) rounds. The set \(U\) returned by the algorithm is an independent set, and the high degree vertices deactivated in each iteration have a node in \(U\) that is at most \(2\) hops away. The rest of the nodes are in \(N^{+}(U)\), which proves correctness.

In each iteration of the while loop, the set \(Z\) can be found in \(O(1)\) rounds using Theorem~\ref{thm:generalderand}. With \(f = \sqrt{\Delta}\), Theorem~\ref{thm:generalderand} gives us that \(G[Z]\) has at most \(O(m/\Delta) \le O(n)\) edges. Therefore, \(G[Z]\) can be sent to \(w\) in constant rounds in linear memory MPC (and also in \cc{} using Lenzen's routing protocol~\cite{lenzen13_optim_deter_routin_sortin_conges_clique}). The rest of the steps in the while loop also require constant rounds.

In each iteration, the maximum degree of the active graph falls from \(\Delta\) to \(\Delta^{0.5 + \eps} \le \Delta^{5/6}\) for \(\eps = 1/3\). This implies there can be at most \(O(\log \log \Delta)\) iterations before \(\Delta\) becomes small enough and we exit the while loop.

Finally the deterministic low-degree MIS algorithm of Censor-Hillel, Parter, and Schwartzman~\cite{Censor-HillelPS16} takes \(O(\log \Delta)\) rounds in the \cc{} model when \(\Delta = O(n^{1/3})\). This algorithm also works in the linear memory MPC model with the same round complexity. In our case \(\Delta = O(\log^{4} n)\) so this algorithm can be used, and it requires \(O(\log \log n)\) rounds, which proves the theorem.

One can alternatively use the deterministic MIS algorithm of Czumaj, Davies, and Parter \cite{Czumaj2019} that takes \(O(\log \Delta + \log \log n)\) rounds in the low-memory MPC model. This algorithm will work with the same round complexity in the linear memory MPC and \cc{} models. So the overall running time of Algorithm~\ref{alg:2rsmpc} remains \(O(\log \log n)\) rounds.
\end{proof}

\section{Future Work}\label{sec:conclusion}
The natural direction suggested by our work is whether it can be applied to obtain fast, deterministic 2-ruling set algorithms in the other two MPC memory regimes? Specifically,
it will be interesting to determine if our techniques also yield a $O(\poly \log\log n)$-round 2-ruling set deterministic algorithm in the sublinear memory MPC model. In the superlinear memory MPC model, Harvey, Liaw, and Liu~\cite{HarveyLiawLiuSPAA2018} give a randomized $O(1)$-round MIS algorithm. It would be interesting to see if one can obtain deterministic $O(1)$-round algorithms for MIS and $2$-ruling set in the superlinear memory MPC model.



\paragraph*{Acknowledgement:}
We thank the anonymous reviewer who suggested a simplification to our approach that also led to an improvement in the running time of Algorithm~\ref{alg:2rsmpc}.

\bibliographystyle{alpha}
\bibliography{short_refs}

\end{document}